\documentclass[12pt]{article}
\usepackage[utf8]{inputenc}
\usepackage{euscript}
\usepackage{a4wide}
\usepackage[T2A]{fontenc}
\usepackage{amsfonts}
\usepackage{amscd}
\usepackage[matrix,arrow,curve]{xy}
\usepackage{amssymb, amsthm}
\usepackage{amsmath}
\usepackage{pdfpages}
\usepackage{graphicx}
\usepackage{ dsfont }
\usepackage{geometry}
\usepackage[unicode, pdftex]{hyperref}
\usepackage{tikz}
\usetikzlibrary{graphs}
\usetikzlibrary{calc}

\usepackage{comment}

\theoremstyle{plain}
\newtheorem{thm}{Theorem} 
\newtheorem{lm}{Lemma}
\newtheorem{st}{Proposition}
\newtheorem*{ex}{Example}

\theoremstyle{definition}
\newtheorem{remark}{Remark}
\newtheorem{defn}{Definition}
\newtheorem*{n}{Remark}

\newtheorem{proposition}{Proposition}

\theoremstyle{remark}

\usepackage{graphicx}
\graphicspath {{./images/}}
\DeclareGraphicsExtensions{.pdf}

\title{WELLDOC property for words generated by morphisms}

\author{Svetlana Puzynina, Vladimir Schavelev\\
Saint Petersburg State University, Russia\\
s.puzynina@gmail.com, vovashavelev11@mail.ru}

\date{}

\begin{document}

\sloppy

\maketitle

\begin{abstract} In this paper, we study an abelian-type property of infinite words called \emph{well distributed occurrences}, or WELLDOC for short. An infinite word $w$ on a $d$-ary alphabet has the WELLDOC property if, for each factor $u$ of $w$, positive integer $m$, and vector 
$v\in \mathbb{N}^d$, there is an occurrence of $u$ such that the Parikh vector of the prefix of $w$ preceding such occurrence is congruent to $v$ modulo $m$. The Parikh vector of a finite word $v$ on an alphabet 
 has its $i$-th component equal to the number of occurrences of the $i$-th letter in $v$. We provide a criterion of the WELLDOC property for words generated by morphisms.
\end{abstract}


\section{Introduction}

In this paper, we study the well distributed occurrences,
or WELLDOC for short, of infinite words generated by morphisms. This is an abelian type property of infinite words regarding the regularity of distribution of factors in a word. An infinite word $u$ over an alphabet $A$ satisfies the WELLDOC property if for each integer $m$ and each
factor $w$ of $u$, the set of Parikh vectors modulo $m$ of prefixes of $u$ preceding the occurrences of $w$ coincides with $\mathbb{Z}_m$ (see Definition \ref{defWD}). The WELLDOC property was originally introduced in the context of pseudorandom number generators \cite{Apng}. 

A {\it linear congruential generator} is a sequence $Z_{n+1} = aZ_n + c \mod m$, for which $a, c, m \in \mathbb{N}$. Linear congruential generators have a defect called lattice structure: if we consider the set of all $n$ consecutive numbers from the generator as a subset of  $\mathbb{Z}^n$, they are covered with a family of equidistant hyperplanes which do not cover the whole plane. Guimond et al. \cite{GuPaPa1} proved that when two linear congruential generators are combined
using infinite words coding certain classes of quasicrystals or, equivalently,
of cut-and-project sets, the resulting sequence is aperiodic and has no lattice
structure. More generally, given an infinite word $w=w_0w_1\cdots$ over an alphabet $\Sigma$ and linear congruential generators $Z^{(i)}$, $i=1,\dots, |\Sigma|$, one can combine them as follows. Consider a function $f$ defined by $f(i) = |\{j < i| w_j = w_i\}|$; i.e., the function $f$ counts the number of letters equal to $w_i$ in the prefix of length $i$ of $w$. Then the generator built from $Z^{(i)}$ according to the word $w$ is defined by  $Z(w)_{n} = Z^{(w_n)}_{f(n)}$.

Choosing the word $w$, one can get a sequence which does not have the lattice structure defect. In \cite{Apng} the authors have found a combinatorial condition of well distributed occurrences
which guarantees absence of the lattice structure in related pseudorandom generators. Besides that, they proved that Sturmian and episturmian words satisfy the WELLDOC property, and proposed an open question:  characterize morphisms generating words satisfying the WELLDOC property. In this paper, we answer the question. The advantage of using words generated by morphisms for building pseudorandom generators is that such words can be generated very fast, so that they do not slow down the generation process.




The main result of the paper is a criterion of the WELLDOC property for words generated by morphisms. For binary alphabets, a recurrent word generated by a morphism has the WELLDOC property if and only if the determinant of the matrix of the morphism equals $\pm 1$ (see Theorem \ref{thdet1bin}). For non-binary words generated by morphisms an additional condition is needed: the WELLDOC property holds if and only if the determinant of the matrix of the morphism is equal to $\pm 1$ and Parikh vectors of returns to the first letter of this word generate the space $\mathbb{Z}^{|\Sigma|}$ (see Theorem \ref{thdet1all}). The latter condition is shown to be decidable. 

The paper is organized as follows. In the next section we provide necessary combinatorial definitions and notation, as well as the statements of our main results. In   
 \hyperlink{alg}{Section \ref{alg}}, we provide necessary backgrounds on algebra. In \hyperlink{secdet1}{Section \ref{secdet1}} we prove that the property from Theorems \ref{thdet1bin} and \ref{thdet1all} is sufficient for the WELLDOC property, and in \hyperlink{secdet0}{Section \ref{secdet0}} we prove that it is necessary.

\section{Definitions and main results}

An \emph{alphabet} $\Sigma$ is a finite set, and its elements are called  \emph{letters}. A  \emph{finite }(resp.,  \emph{infinite})  \emph{word} is a finite (resp.,  infinite) sequence $w = w_0w_1\cdots$ of elements from $\Sigma$. We let $\Sigma^*$ denote the set of finite words. The empty word is denoted by $\varepsilon$. A \emph{factor} of $w$ is a finite subsequence of its consecutive letters. A \emph{prefix} is a particular case of a factor starting from position $0$. For a finite word $u=u_0\cdots u_{n-1}$, its \emph{suffix} is a factor ending at position $n-1$, and its \emph{proper factor} is a factor that is neither its prefix, nor its suffix, nor the empty word. For a finite word $u=u_0\cdots u_{n-1}$, its length is denoted by $|u|=n$, and the number of occurrences of a letter $a$ in it is denoted by $|u|_a$. For an infinite word $w = w_0w_1\cdots$, we let $w[i, i+k)$ denote its factor of length $k$ starting from the position $i$:  $w[i, i+k)=w_i\cdots w_{i+k-1}$. 

A \emph{morphism} $\phi$ is a map from $\Sigma^*$ to $\Sigma^*$ such that $\phi(uv) = \phi(u)\phi(v)$ for all finite words $u,v \in \Sigma^*$. All morphisms considered in this paper are \emph{nonerasing}, i.e., the image of any non-empty word is non-empty. A morphism $\phi$ is prolongable on a letter $a$ if $\phi(a) = as$ for some nonempty finite word $s$. For each letter $a$ on which $\phi$ is prolongable, the morphism generates an infinite word beginning with $a$: $w = \lim_{n \to \infty} \phi^n(a),$ where the limit is considered in the prefix sense. Such words are also called \emph{fixed points} of the morphism $\phi$, since $w = \phi(w)$.  They are also referred to as (purely) morphic words and D0L-words. Note that we have the following equality: $w=as\phi(s)\phi^2(s)\cdots\phi^{(n)}(s)\cdots$. 
     A morphism $\phi$ is called \emph{primitive} if for each letter $x$ and some $k$ the word $\phi^k(x)$ contains each letter of the alphabet, and non-primitive otherwise. 

In the further text, we take the alphabet to be the initial segment of integers: $\Sigma = \{0, 1, \ldots, \sigma-1\}$ for some integer $\sigma$; so, $\sigma =|\Sigma|$. Without loss of generality we will consider only fixed points of morphisms beginning with $0$, and so we suppose that the morphisms we consider are prolongable on $0$. 

\begin{defn}
    For a finite word $u$, its {\it Parikh vector} $V_u$ is the vector $V_u=(|u|_0, |u|_1, \ldots, |u|_{\sigma-1})$.
\end{defn}

In the further text, we will write all vectors horizontally.
 
\begin{defn}
     Let $w$ be an infinite word, $u$ its factor and   $a_0, a_1, \ldots$ the set of all positions of occurrences of $u$, i.e. $u=w[a_i, a_i+|u|)$. Words of the form $w[0,a_i)$ are then called 
 {\it prefixes preceding $u$}. We let $X_u$ denote the set of Parikh vectors of all prefixes preceding $u$ in $w$, i.e., $X_u=\{ V_{w[0,a_i)} | i\in \mathbb{N}\}$. We let $X_{u, m}$ denote the set of vectors over $\mathbb{Z}/m\mathbb{Z}$  from $X_u$ modulo $m$.  
    
\end{defn}

In \cite{Apng} the authors found the following condition called WELLDOC for a word $w$ which guarantees absence of the lattice structure defect for the sequence  $Z(w)_n$:

\hypertarget{defWD}{}
\begin{defn} \label{defWD}
An infinite word $w$ satisfies the {\it WELLDOC property} if for each integer $m$ and each factor $u$ of $w$ we have $X_{u,m} = (\mathbb{Z}/m\mathbb{Z})^\sigma$, i.e., for each vector  $v=(v_0, \ldots, v_{\sigma-1})$ there exists a prefix $p$ preceding $u$ such that $|p|_i \equiv_m v_i$ for each $i$, or, equivalently, $V_p \equiv_m v$. 
\end{defn}

As the main result of this paper, we find a criterion of the WELLDOC property for words generated by morphisms. To formulate this criterion, we need the following definition: 

\begin{defn}

    A \emph{matrix} of a morphism $\phi$ is the following matrix of order $\sigma=|\Sigma|$: 
    
    \begin{equation*}
        A_\phi = 
        \begin{pmatrix}
            |\phi(0)|_0 & |\phi(1)|_0 & \ldots & |\phi(\sigma-1)|_0 \\
            |\phi(0)|_1 & |\phi(1)|_1 & \ldots & |\phi(\sigma-1)|_1 \\
            \vdots & \vdots & \ddots & \vdots \\
            |\phi(0)|_{\sigma-1} & |\phi(1)|_{\sigma-1} & \ldots & |\phi(\sigma-1)|_{\sigma-1} \\
        \end{pmatrix}
    \end{equation*}

    
    
\end{defn}

\hypertarget{thdet1bin}{}
\begin{thm} \label{thdet1bin} 
    Let $w$ be an infinite recurrent binary word generated by a morphism $\phi$. Then $w$ satisfies  the WELLDOC property if and only if $\det A_\phi = \pm 1$.
\end{thm}

\begin{remark}
    In this paper we do not consider the words $0^\infty$ and $1^\infty$ to be binary words. Clearly, for these words the WELLDOC property holds, although they are defined by the morphism $\phi(0) = 00, \phi(1) = 11$, while its matrix has determinant equal to $4$. 
\end{remark}

\begin{remark}   Notice that only recurrent words can satisfy the WELLDOC property by definition. It is easy to check if a purely morphic word is recurrent. Indeed, consider increasing powers of the morphism until the set of letters in the image of each letter stays the same (we need to consider exponents at most $|\Sigma|^2$). If $0$ is in the image of some letter $a$ of this power and $a$ is in the image of $0$, then the fixed point of this morphism beginning with 0 is recurrent, otherwise it is not.  So, we state all our results for recurrent words. 
\end{remark}

\hyperlink{thdet1bin}{Theorem \ref{thdet1bin}} will be proved in the following two sections: the sufficiency of the condition for the determinant in \hyperlink{thdet1>wbin}{Theorem \ref{thdet1>wbin}}, and the necessity in Theorem \ref{th:necessary}. 

Before proving  Theorem \ref{thdet1bin}, we provide an example of its application to Sturmian words. More precisely, we will reprove in part a result from  \cite{Apng}.

\begin{defn}
    An infinite word $w$ is called {\emph{Sturmian}} if for each $n \in \mathbb{N}$ the word $w$ contains exactly  $n+1$ factors of length $n$.
\end{defn}

\begin{defn}
    A morphism $\phi$ is called \emph{Sturmian} if  $\phi(w)$ is a Sturmian word for each Sturmian word $w$. 
\end{defn}

Some (but not all) Sturmian words are generated by morphisms, and it is well known that these morphisms are Sturmian morphisms  (see, e.g., \cite[Chapter 2]{W}). 

In \cite{Apng}, the authors prove that Sturmian words satisfy the WELLDOC property; a particular case of this result for Sturmian words generated by morphisms is a corollary from  \hyperlink{thdet1bin}{Theorem \ref{thdet1bin}}: 

\begin{proposition}
    Sturmian words generated by morphisms satisfy the WELLDOC property.
\end{proposition}
\begin{proof}
    We will make use of the following fact (see, e.g., \cite[Chapter 2.3]{W}): each Sturmian morphism is a composition of several morphisms of the following form: 
    \begin{equation*}
        E: \begin{matrix}
        0 \to 1 \\ 1 \to 0
        \end{matrix} \qquad
        \psi_0: \begin{matrix}
        0 \to 01 \\ 1 \to 1
        \end{matrix} \qquad
        \widetilde{\psi}_0: \begin{matrix}
        0 \to 10 \\ 1 \to 1
        \end{matrix}
    \end{equation*}
Note that matrices of these morphisms have determinant $\pm 1$, so the matrix of any their composition also has determinant $\pm 1$. Thus, we have the WELLDOC property by  \hyperlink{thdet1bin}{Theorem \ref{thdet1bin}}.
\end{proof}

Theorem \ref{thdet1bin} can be generalized to nonbinary alphabets with an additional condition:

\begin{thm} \label{thdet1all}
    Let $w$ be an infinite word generated by a morphism $\phi$. Then $w$ satisfies the WELLDOC property if and only if $\det A_\phi = \pm 1$ and Parikh vectors of all returns to the first letter of $w$ generate $\mathbb{Z}^{|\Sigma|}$ as an additive group.
\end{thm}


We prove \hyperlink{thdet1bin}{Theorem \ref{thdet1all}} in the following two sections: the sufficiency of the condition on the determinant in  
\hyperlink{thdet1}{Theorem \ref{thdet1}}, and the necessity in Theorem \ref{th:necessary}.

\hypertarget{alg}{}

\section{Algebraic statements} \label{alg}

Here we gather all algebraic facts used in this paper. All statements, probably with the exception of the last one, are well-known and can be found in textbooks on algebra, e.g.,  \cite{AM, Alin,  A}. We provide references or proofs for all the statements for the sake of completeness.
\hypertarget{algdet1}{}
\begin{proposition} \label{algdet1} 
    A matrix $A$ over a field is invertible if and only if $\det A \ne 0$. A matrix $A$ over $\mathbb{Z}$ is invertible if and only if $\det A = \pm 1$. This in turn is equivalent to the fact that the matrix $A$ with entries modulo $m$ is invertible over $\mathbb{Z}/m\mathbb{Z}$ for each $m$.
\end{proposition}

\begin{proof}
    The first part is \cite[Chapter VII, Theorem 5.1]{Alin}, the second one is \cite[Excercise 2.109]{A}. Let us prove the third one. If $A$ is invertible over $\mathbb{Z}$, then its inverse matrix modulo $m$ is the inverse matrix to $A \mod m$, hence $A \mod m$ is invertible. If $\det A \ne \pm 1$, then there exists a prime divisor $p$ of $\det A$, and hence $\det A =0 \mod p$. Therefore, by the first statement of the proposition, $A$ is not invertible.
\end{proof}

The following proposition is a standard corollary of Lagrange's theorem for groups \cite[Theorem 4.50 and Corollary 4.69]{A}. 
\hypertarget{algfingr}{} 
\begin{proposition} \label{algfingr} 
    Let $G$ be a finite group, then for each $g \in G$ there exists $n$ such that  $g^n=e$, and  $|G|$ is divisible by  $n$. 
\end{proposition}

\hypertarget{algbas=det1}{} 
\begin{proposition} \label{algbas=det1}
    Let $A_1, \ldots, A_n$ be vectors of dimension $n$ over a field. Then they are linearly dependent if and only if the determinant of the matrix constructed from these vectors is equal to 0. If the vectors have integer entries, then they form a spanning set of $\mathbb{Z}^n$ if the determinant of this matrix is $\pm 1$. %
   
\end{proposition}

\begin{proof}
    The first part is  \cite[ch. 7, Theorem 3.1]{Alin}. We now prove the second part.
    
    By \hyperlink{algdet1}{Proposition \ref{algdet1}}, this matrix is invertible. Now we are going to prove that each vector $v \in \mathbb{Z}^n$ can be 
    expressed as a linear combination of the vectors $A_i$. Note that $A\cdot A^{-1}\cdot v = v$, and we denote $(u_1,\ldots, u_n) = u = A^{-1}v$. Then on one hand $Au = v$, and on the other hand $Au = \sum \limits_{1 \le i \le n} u_iA_i$. This gives a presentation of $v$ as a linear combination of the vectors  $A_i$.
\end{proof}

\hypertarget{algdet0}{} 
\begin{proposition} \label{algdet0}
    Let $A$ be a matrix $n \times n$ over $\mathbb{Z}$ with $\det A \ne \pm 1$. 
    Then $A(\mathbb{Z}/m\mathbb{Z})^n\neq (\mathbb{Z}/m\mathbb{Z})^n$ for some integer $m$. In other words, there exists a vector $v \in (\mathbb{Z}/m\mathbb{Z})^n$ such that $Au \ne v$ for each $u \in (\mathbb{Z}/m\mathbb{Z})^n$. 
\end{proposition}

\begin{proof}
    Assume that $\det A =k$; we let $p$ be a prime divisor of $k$. Consider the matrix $A$ with its entries taken modulo $p$ as a matrix over the field  $\mathbb{Z}/p\mathbb{Z}$. The determinant of the new matrix is equal to  0, so, by \hyperlink{algbas=det1}{Proposition \ref{algbas=det1}}, the vectors of the matrix are linearly dependent.  So, since we had $n$ vectors of dimension $n$, they do not form a spanning set, so we can indeed find such a vector $v$. Therefore, we can choose $m=p$ (in fact, in the statement we can choose $m$ to be prime, although we do not really need it).
\end{proof}

We will make use of the following proposition, which is less classical than the previous statements from this section. Given a set of vectors $V$ and an integer $m$, we let $V/mV$ denote this set of vectors modulo $m$, i.e., we take each coordinate of a vector modulo $m$. 

\hypertarget{algn=allp}{}
\begin{proposition} \label{algn=allp}
    Let $V$ be a set of vectors from $\mathbb{Z}^n$.
    Then $V$ generates $\mathbb{Z}^n$ as an additive group if and only if $V/pV$ generates $(\mathbb{Z}/p\mathbb{Z})^n$ for each prime $p$.
\end{proposition}

\begin{proof}
    If $V$ generates $\mathbb{Z}^n$, then, in particular, a unit vector $e_i = (0, \ldots, 1, \ldots, 0)$ with  $i$-coordinate equal to $1$ and all other coordinates equal to $0$ can be expressed as a linear combination of vectors from $V$. In other words, there exist integers $a_{i,j}$ such that $e_i=a_{i,0}v_0+\ldots+a_{i,n-1}v_{k-1}$, where $v_0, \ldots, v_{k-1}$ are vectors from $V$ for some $k \in \mathbb{N}$. 
    Then for each modulo $p$ the same linear combination $a_{i,0}v_0+\ldots+a_{i,k-1}v_{k-1}$ 
    gives the same vector $e_i$ modulo $p$. 
    So, for each $p$ and each $e_i$  we can express $e_i$ as a linear combination of vectors from  $V/pV$, and hence we can generate $(\mathbb{Z}/p\mathbb{Z})^n$.
    
    Suppose now that $V/pV$ generates $(\mathbb{Z}/p\mathbb{Z})^n$ for each integer $p$. We let $G$ denote the subgroup of $\mathbb{Z}^n$ generated by $V$. Note that the group generated by $V/pV$ is equal to $G/pG$, hence $G/pG = (\mathbb{Z}/p\mathbb{Z})^n$ 
    and the factor group $(\mathbb{Z}/p\mathbb{Z})^n/(G/pG)$ is isomorphic to the trivial group $0$. 
    
    Consider the factor group $H = \mathbb{Z}^n/G$. Note that $H$ is abelian finitely generated group, so the classification of abelian finitely generated groups \cite[Theorem 9.28]{A} implies that
    $H$ can be expressed in the following form: \begin{equation}\label{eq:H} H \cong \mathbb{Z}^k \oplus \mathbb{Z}/p_1^{\alpha_1}\mathbb{Z}\oplus \dots \oplus \mathbb{Z}/p_l^{\alpha_l}\mathbb{Z}\end{equation} for some $k, \alpha_1, \ldots, \alpha_l \in \mathbb{N}$ 
    and some prime numbers $p_1, \ldots, p_l$ ($p_i$'s do not have to be pairwise distinct). We will now prove that $H \cong 0$, which implies, in particular, that $G = \mathbb{Z}^n$. Assume the converse: Suppose that $H \ncong 0$. 
    
    We will make use of the following claim:  
    
    \bigskip
    
    \textit{Claim.} Let $B \le A, D \le C$ be abelian groups and their subgroups. Then   $$(A \oplus C)/(B \oplus D) \cong A / B \oplus C / D.$$
    
    The claim is proved by building a natural isomorphism mapping a class $(a,c) + B \oplus D$ to a direct sum of the classes $(a+B) \oplus (c+D)$.

    \bigskip
    
    Our current goal is to prove that $H/pH \ncong 0$ for some prime $p$. Due to \eqref{eq:H}, we have that $pH \cong p\mathbb{Z}^k \oplus p(\mathbb{Z}/p_1^{\alpha_1}\mathbb{Z})\oplus \dots \oplus p(\mathbb{Z}/p_l^{\alpha_l}\mathbb{Z})$. Now, to describe $H/pH$, it is enough to describe $\mathbb{Z}^k/p\mathbb{Z}^k$ and $(\mathbb{Z}/p_i^{\alpha_i}\mathbb{Z})/p(\mathbb{Z}/p_i^{\alpha_i}\mathbb{Z})$, using the isomorphism described above. 
    Using this isomorphism and the fact that  $\mathbb{Z}^k = \mathbb{Z} \oplus \dots \oplus \mathbb{Z}$ we obtain that $\mathbb{Z}^k/p\mathbb{Z}^k \cong (\mathbb{Z}/p\mathbb{Z})^k$. Summands of the form $p(\mathbb{Z}/p_i^{\alpha_i}\mathbb{Z})$ in $pH$ can be rewritten as follows:
    
    \begin{equation*}
        p(\mathbb{Z}/p_i^{\alpha_i}\mathbb{Z}) = \begin{cases}
            \mathbb{Z}/p_i^{\alpha_i}\mathbb{Z}, \mbox{ if $p \ne p_i$,} \\ \mathbb{Z}/p_i^{\alpha_i-1}\mathbb{Z}, \mbox{ if $p = p_i$.}
        \end{cases}
    \end{equation*}
    Then $(\mathbb{Z}/p_i^{\alpha_i}\mathbb{Z})/p(\mathbb{Z}/p_i^{\alpha_i}\mathbb{Z}) \cong 0$, if $p_i \ne p$, and $(\mathbb{Z}/p_i^{\alpha_i}\mathbb{Z})/p(\mathbb{Z}/p_i^{\alpha_i}\mathbb{Z}) \cong \mathbb{Z}/p\mathbb{Z}$, if $p_i = p$.
    
    If $k \ne 0$, then $(\mathbb{Z}/p\mathbb{Z})^k$ is a summand in $H/pH$ and it is not isomorphic to $0$, so $H/pH \ncong 0$. If $l \ne 0$, then we choose as $p$ one of the numbers $p_i$, hence $H/pH$ contains a summand $\mathbb{Z}/p\mathbb{Z} \ncong 0$, and thus $H/pH \ncong 0$. If $k=l=0$, then $H \cong 0$, which we assumed to be false.
    
    To obtain a contradiction, it remains to prove that the factor group $(\mathbb{Z}/p\mathbb{Z})^n/(G/pG)$ is isomorphic to the factor group $H/pH$, since we already showed that the first one is isomorphic to 0, while the second one is not.
    
    We now let $A$ denote $\mathbb{Z}^n$, $B$ denote $G$, $C$ denote $p\mathbb{Z}^n$ and $D$ denote $pG$. Since  $(\mathbb{Z}/p\mathbb{Z})^n \cong \mathbb{Z}^n/p\mathbb{Z}^n=A/C$ and $pH \cong p\mathbb{Z}^n/pG = C/D$, 
    it remains to verify that $(A/C)/(B/D) \cong (A/B)/(C/D)$, and $D = B \cap C$. For that, we make use of the following theorems about isomorphisms \cite[Proposition 2.1]{AM}: 
    
    For any subgroups $H'$ and $K'$ of a group $G'$ the following properties hold: 
    \begin{enumerate}
        \item $H'K' / H' \cong K'/(K' \cap H')$,
        \item If $H' \unlhd K'$, then $ (G'/K') / (H'/K') \cong G'/H'$.
    \end{enumerate}
    
    Now it is easy to prove that our groups are isomorphic: 
    
    $$(A/C)/(B/(B \cap C)) \cong (A/C)/(BC/C) \cong A/BC \cong (A/B)/(BC/B) \cong (A/B)/(C/(B \cap C)),$$
    where the first and the fourth isomorphisms follow from Property 1, and the second and the third isomorphisms follow from Property 2.
\end{proof}

\hypertarget{secdet1}{}
\section{Sufficiency of the condition 
for the WELLDOC property} \label{secdet1}

In this section, we prove the sufficiency of the condition for the WELLDOC property from Theorem \ref{thdet1bin}. More precisely, we prove the sufficiency of the condition on the determinant of the matrix of the morphism in the binary case (Theorem \ref{thdet1>wbin}) and its generalization for a non-binary case  (Theorem \ref{thdet1}). The proof is provided in Subsection \ref{subsec:suf_proof}; we then discuss as an example morphic episturmian words in Subsection \ref{subsec:suf_ex}, and finally we prove the decidability of the condition in Subsection \ref{subsec:suf_dec}.

\subsection{Proof of sufficiency of the condition 
for the WELLDOC property} \label{subsec:suf_proof}

Recall that we fixed the alphabet to be  $\Sigma = \{0, 1, \ldots, \sigma-1 \}$ and $|\Sigma| = \sigma$.

\begin{defn}
    An infinite word $w$ satisfies {\it WELLDOC for 0} if for each integer $m$ the set $X_0$ satisfies $X_{0,m} = (\mathbb{Z}/m\mathbb{Z})^\sigma$. 
\end{defn}

In other words, we check the WELLDOC property only for one factor $0$. The following lemma states that for words generated by morphisms with $\det A_\phi = \pm 1$ this condition is equivalent to the WELLDOC property.

\hypertarget{lmonly0}{}

\begin{lm} \label{lmonly0}
    Let $w$ be an infinite word beginning with $0$ generated by a morphism $\phi$ with matrix $A_\phi$ such that $\det A_\phi = \pm 1$. Then for the word $w$ the WELLDOC property is equivalent to the WELLDOC for $0$ property.
\end{lm}

\begin{proof}
    Clearly, WELLDOC implies WELLDOC for 0; we need to show the converse. Consider an arbitrary factor $u$ of $w$; we let $k$ denote the position of the end of its first occurrence. Consider an integer $n$ such that $|\phi^n(0)| > k$, then $u$ is a factor of  $\phi^n(0)$. 
    Given $m$, we now prove by induction on $n$ the following fact: the WELLDOC property is satisfied for each factor $v$ such that its first occurrence is contained in $\phi^n(0)$. The base of induction is for $n=0$, and this is exactly WELLDOC for 0.
    
    Suppose that we have WELLDOC for all factors of $\phi^n(0)$, in particular, for $\phi^n(0)$ itself; we need to prove it for a factor $u$ of $\phi^{n+1}(0)$, i.e.,  $\phi^{n+1}(0)=xuy$ for some $x,y \in \Sigma^*$. Let $s = (s_0, \ldots, s_{\sigma-1})$ be an arbitrary vector from $(\mathbb{Z}/m\mathbb{Z})^{\sigma}$. To prove that it is in $X_{u,m}$, we need to find a prefix preceding $u$ with such Parikh vector $V$
    modulo $m$.
    By induction hypothesis, we have that for each vector  $t = (t_0, \ldots, t_{\sigma-1}) \in (\mathbb{Z}/m\mathbb{Z})^{\sigma}$ there exists a prefix $p$ preceding $\phi^n(0)$ with its Parikh vector $V_p \equiv_m t$. Consider the word $\phi(p)$: it is a prefix of $w$ with the following Parikh vector: $V_{\phi(p)} \equiv_mA_\phi t$. This prefix precedes the factor $\phi(\phi^n(0))=\phi^{n+1}(0)$. 
    
    Note that $\phi(p\phi^n(0))=\phi(p)\phi^{n+1}(0)=\phi(p)xuy$, so $\phi(p)x$ is a prefix preceding $u$. Its Parikh vector is  $V_{\phi(p)x} \equiv_m A_\phi t + V_x$. We can obtain this way any vector $s$ modulo $m$; for that we should take $t = A^{-1}_\phi (s-V_x)$. The inverse matrix is an integer matrix by  \hyperlink{algdet1}{Proposition \ref{algdet1}}, since the determinant of $A_{\phi}$ is equal to  $\pm 1$ for each $m$.
\end{proof}

To prove the sufficiency of the condition in  \hyperlink{thdet1}{Theorem \ref{thdet1bin}}, it remains to prove that $X_{0,m}$ coincides with  $(\mathbb{Z}/m\mathbb{Z})^\sigma$. We first prove some properties of the set $X_0$.

\hypertarget{lmsubgr}{}
\begin{lm} \label{lmsubgr}
    For a recurrent infinite word $w$ generated by a morphism and beginning with $0$, the set $X_{0,m}$ is a subgroup of $(\mathbb{Z}/m\mathbb{Z})^\sigma$.
    In particular, if we have  $(s_0, \ldots, s_{\sigma-1}), (t_0, \ldots, t_{\sigma-1}) \in X_{0,m}$, then $(s_0+t_0, \ldots, s_{\sigma-1}+t_{\sigma-1}) \in X_{0,m}$.
\end{lm}

\begin{proof}
    To prove that $X_{0,m}$ is a subgroup of $(\mathbb{Z}/m\mathbb{Z})^\sigma$, we need to prove its closure under addition and taking the inverse. Note that the closure under the inverse is a consequence of the closure under addition: the inverse of an element $v \in X_{0,m}$ in $(\mathbb{Z}/m\mathbb{Z})^\sigma$ is $(m-1)v$, which can be represented as a sum  of elements from $X_{0,m}$: $(m-1)v = \underbrace{v+\ldots+v}_{m-1 \text{ times}}$.

    Let $p0$ and $q0$ be prefixes of $w$. We need to find a prefix $r0$ such that $V_r \equiv_m V_p+V_q$. Take an integer $k$ such that $q$ is a factor of $\phi^k(0)$. For that we will use several basic algebraic facts. 
    
    Consider the group $GL(\sigma, \mathbb{Z}/m\mathbb{Z})$ of invertible   $\sigma\times\sigma$-matrices modulo $m$. It is finite, so we can apply \hyperlink{algfingr}{Proposition \ref{algfingr}}; it follows that there exists an integer $n$ such that $A_\phi^n \equiv_m I_{\sigma}$. Then 
     $V_{\phi^n(p)} \equiv_m V_p$
     , and if $p0$ is a prefix, then $\phi^n(p)0$ is also a prefix, since $\phi(0)$ begins with the letter $0$. So, $\phi^{nk}(p0)=\phi^{nk}(p)\phi^{nk}(0)$. Here 
    $\phi^{nk}(p)$ has a Parikh vector equivalent to $V_p$ modulo $m$, and $\phi^{nk}(0)$ starts with $q0$, since $nk > k$ and $q0 $ is a prefix of  $\phi^k(0)$. Thus,  $\phi^{nk}(p)q0$ is a prefix with the vector equivalent to $V_p+V_q$ modulo $m$. 
    \end{proof}

For the next lemma we will need a definition of a return word (see, e.g., \cite{return}).

\begin{defn}
    Let $w$ be an infinite word, $u$ its factor, and $a_1\leq a_2 \leq \ldots$ be the sequence of all positions of occurrences of $u$ in $w$. Then factors $w[a_i,a_{i+1})$ are called {\it returns to $u$}.
\end{defn}

\begin{lm} \label{lmX0ret}
    For an infinite word $w$ generated by a morphism and beginning with $0$, the set $X_0$ is generated as an additive group by the set of vectors $$\{V_{x} \mid x\mbox{ is a return to 0 in } w\}.$$
    
\end{lm}


\begin{proof}
    Indeed, $X_0 = \{V_p : p\mbox{ is a prefix of } w \mbox{ preceding 0}\}$, and each of these prefixes can be factorized into returns to 0. Thus, $V_p = V_{x_{i_1}}+\ldots+V_{x_{i_k}}$. 
\end{proof}

Now we will show that it is enough to prove  WELLDOC only for prime moduli.

\begin{lm} \label{lmonlyp}
    For an infinite word $w$ generated by a morphism and beginning with $0$, the condition $X_{0,m} = (\mathbb{Z}/m\mathbb{Z})^\sigma$ for each $m$ is equivalent to the same condition $X_{0,m} =  (\mathbb{Z}/m\mathbb{Z})^\sigma$ only for prime numbers $m$.
\end{lm}

\begin{proof}
    Clearly, if the condition holds for each integer $m$, then it holds in particular for prime numbers $m$. To prove the converse, it is enough to prove that if  $X_{0,m}=(\mathbb{Z}/m\mathbb{Z})^\sigma$ and $X_{0,p}= (\mathbb{Z}/p\mathbb{Z})^\sigma$, then $X_{0,mp}= (\mathbb{Z}/mp\mathbb{Z})^\sigma$, where  $m$ is an integer and $p$ is prime.
    
Since $X_{0,p} = (\mathbb{Z}/p\mathbb{Z})^\sigma$, there exists $v_i\in X_0$ such that $v_i \!\!\mod p$ has $1$ only in the $i$-th coordinate and $0$'s in all other coordinates, i.e., it is equal to $(0,0, \ldots, 1, \ldots, 0)$. 
If we consider it modulo $mp$, it is equal to $(pa_0, pa_1, \ldots, 1+pa_i, \ldots, pa_{\sigma-1})$ for some integers $a_i$. 
    \hyperlink{lmsubgr}{Lemma \ref{lmsubgr}} implies that $X_{0,mp}$ is a subgroup, 
    hence $mv_i=\underbrace{v_i+\ldots+v_i}_{m \text{  times}}\in X_{0,mp}$
    , whereas it is equal to $(0,0, \ldots, m, \ldots, 0)$. In the same way, considering vectors modulo $m$, we obtain all vectors of the form $(0,0, \ldots, p, \ldots, 0)$. Now we can subtract from each vector $(pa_0, pa_1, \ldots, 1+pa_i, \ldots, pa_{\sigma-1})$ all $pa_i$'s and obtain all vectors $(0,0, \ldots, 1, \ldots, 0)$ modulo $mp$. 
\end{proof}


Now we can give a sufficient condition for the WELLDOC property for recurrent words generated by morphisms (this is in fact part  of the statement of Theorem~\ref{thdet1all}). 
 

\hypertarget{thdet1}{}

\begin{thm} \label{thdet1}
    Let $w$ be a recurrent infinite word generated by a morphism $\phi$ such that $\det A_\phi = \pm 1$. Then $w$ satisfies the WELLDOC property if and only if Parikh vectors of returns to the first letter $w_0$ of the word $w$ generate $\mathbb{Z}^\sigma$ as an additive group. 
\end{thm}

\begin{remark}
The condition on Parikh vectors of returns to 0 can be checked algorithmically (see Lemma \ref{lmret}  and Proposition \ref{pr_alg}.)
\end{remark}

\begin{proof}
    Without loss of generality we assume that $w_0 = 0$. We note that only recurrent words satisfy the WELLDOC property by definition, and for recurrent words we can apply the lemmas above. Lemma \ref{lmonly0} implies that the WELLDOC property is equivalent to the fact that $X_{0,p} =  (\mathbb{Z}/p\mathbb{Z})^\sigma$ for each $p$, and Lemma \ref{lmonlyp} implies that it is enough to check the WELLDOC property only for prime numbers $p$. Lemma \ref{lmX0ret} implies that the set $X_0$ considered as a group with the standard operation of vector addition is generated  
    by Parikh vectors of returns to $0$. It remains to prove that Parikh vectors generate
   the set $\mathbb{Z}^\sigma$ if and only if they generate  $(\mathbb{Z}/p\mathbb{Z})^\sigma$ for each prime $p$, which follows from \hyperlink{algn=allp}{Proposition \ref{algn=allp}}.
\end{proof}

\subsection{Example: episturmian words} \label{subsec:suf_ex}

As an example of the application of  \hyperlink{thdet1}{Theorem \ref{thdet1}}  we will reprove in part 
a result from \cite{Apng} for a subclass of standard episturmian words generated by morphisms. A factor  $u$ of an infinite word $w$ is called \textit{left special} if it can be continued to the left with at least two distinct letters, i.e., there exist letters $a, b \in \Sigma$, $a \ne b$ such that $au$ and $bu$ are also factors of $w$.

\begin{defn}
    An infinite aperiodic word $w$ is called {\it episturmian} if the set of factors of $w$ satisfies the following properties:
    
    \begin{enumerate}
        \item For each factor $u$ the word $u^R$ is also a factor of $w$, where $u^R$ denotes the word $u$ written in the reversed order.
        \item For each $n \in \mathbb{N}$ there exists exactly one left special factor $u$ of length $n$.
    \end{enumerate}

    If in addition all left special factors are prefixes of $w$, then $w$ is called \textit{standard episturmian}.
\end{defn}

\begin{defn}
    A morphism $\phi$ is called {\it episturmian} if  $\phi(w)$ is episturmian for each episturmian word $w$. 
\end{defn}

In \cite{Apng}, the WELLDOC property has been proved for episturmian words. We will now reprove this fact for standard episturmian words generated by morphisms using \hyperlink{thdet}{Theorem \ref{thdet1}}. In fact, it follows that the WELLDOC property holds for episturmian words with the same set of factors as standard episturmian words generated by morphisms, since the WELLDOC property depends only on the set of factors.

\begin{st}
    Standard episturmian words generated by morphisms satisfy the WELLDOC property.
\end{st}

\begin{proof}
    Episturmian morphisms admit an explicit description similar to Sturmian morphisms
    (see \cite{episturm}). A morphism is episturmian if and only if it is a composition of the following morphisms (here $a, b$ and $x$ are letters): 
    
    \begin{equation*}
        \theta_{ab}: \begin{matrix}
        a \to b \\ b \to a \\ x \to x \quad \forall x \in \Sigma\backslash\{a,b\}
        \end{matrix} \quad
        \psi_a: \begin{matrix}
        a \to a \\ x \to ax \quad \forall x \in \Sigma\backslash\{a\}
        \end{matrix} \quad
        \widetilde{\psi}_a: \begin{matrix}
        a \to a \\ x \to xa \quad \forall x \in \Sigma\backslash\{a\}
        \end{matrix}
    \end{equation*}
    
    In other words, $\theta_{ab}$ trades places two letters; the morphisms $\psi_a$ and $\widetilde{\psi}_a$ add the letter $a$ to all other letters. Similarly, for Sturmian morphisms, determinants of matrices of these morphisms are equal to $\pm 1$. Therefore, the matrices of episturmian morphisms also have determinants $\pm 1$, and we can apply  \hyperlink{thdet1}{Theorem \ref{thdet1}}.

    For a given standard episturmian word $w$ consider its first letter $w_0$; without loss of generality we may assume that $w_0=0$. By Theorem 4.4 in \cite{epireturn}, the set of return words to the prefix $0$ is the set $\{\psi_0(a)\mid a \in \Sigma\}$. Therefore, we need to check whether $\{V_0\} \cup \{V_{0a} \mid a \in \Sigma \setminus \{0\}\}$ generates $\mathbb{Z}^\sigma$ as an additive group. Clearly, this is the case, since  $\{V_0\} \cup \{V_{0a} -V_0 \mid a \in \Sigma \setminus \{0\}\}$ is the standard basis of $\mathbb{Z}^\sigma$.
\end{proof}

Now we give an example of a word generated by morphism such that the determinant of the matrix of the morphism is $\pm 1$; however, the WELLDOC property does not hold, which shows that the condition on returns from Theorem \hyperlink{thdet1}{\ref{thdet1}} is indeed necessary in the non-binary case.

\begin{ex}
    Consider the following morphism $\phi$: 
    
    \begin{equation*}
        \phi: \begin{matrix}
        0 \to 02,\mbox{ } \\ 1 \to 101, \\ 2 \to 102;
        \end{matrix} \qquad 
        A_\phi = \begin{pmatrix}
            1 & 1 & 1 \\
            0 & 2 & 1 \\
            1 & 0 & 1 \\
        \end{pmatrix}.
    \end{equation*}
    We have $\det A_\phi = 1$, but the returns to   $0$ are $01$ and $021$, so $X_0 = \left<(1,1,0),(1,1,1)\right>$.
\end{ex}

\subsection{Decidability of the sufficient condition} \label{subsec:suf_dec}

We first provide a proof in the case of primitive morphisms, and then proceed with the general case.

We will need one more definition. 

\begin{defn}
    Let $w$ be an infinite word. We define the \textit{graph of $w$} as a directed graph whose set of vertices is identified with the set of letters (i.e., it has $\sigma$ vertices), and there is an edge from letter $a$ to letter $b$ if and only if $ab$ is a factor of $w$.
\end{defn}

\begin{n}
   This graph is also a Rauzy graph of order 1.
\end{n}

\begin{ex}
    The word $(0101201)^\infty$ has the following graph:
    
    
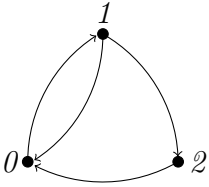
\begin{tikzpicture}
        \coordinate [label=left:0] (0) at (0,0);
        \coordinate [label=above:1] (1) at (1,1.7);
        \coordinate [label=right:2] (2) at (2,0);
        \filldraw[black] (1, 1.7) circle (2pt);
        \filldraw[black] (0, 0) circle (2pt) ;
        \filldraw[black] (2,0) circle (2pt) ;
        
        \draw[black, ->] (0) arc (180:123:2);
        \draw[black, ->] (1) arc (60:3:2);
        \draw[black, ->] (2) arc (-60:-117:2);
        \draw[black, ->] (1) arc (0:-57:2);
        
    \end{tikzpicture}
    
\end{ex}

 
The following lemma provides an algorithm to compute returns to letters in primitive morphisms. This implies that the sufficient condition in  \hyperlink{thdet1}{Theorem \ref{thdet1}} can be checked algorithmically.
 
\hypertarget{lmret}{}

\begin{lm} \label{lmret}
    Let $w$ be an infinite word generated by a primitive morphism $\phi$, and suppose that for each letter $a$ its image $\phi(a)$ contains all letters from the alphabet $\Sigma$. For each letter $a$ we decompose its image in the form  $\phi(a) = x_a0v_{a,0}0v_{a,1}0\cdots 0v_{a,k_a-1}0y_a$, where $v_{a,i}, x_a, y_a \in (\Sigma\setminus 0)^*$.   Then each return to $0$ in $w$ is either of the form $0v_{a,i}$, or of the form $0y_ax_b$ for some letters $a,b$, if the graph of $w$ has the edge $ab$.
\end{lm} 

\begin{proof}
    Consider a factorization of $w$ into blocks $\phi(a)$ for each letter $a$. Note that each return of $w$ overlaps at most two blocks, since each block contains at least one occurrence of 0. Each return either lies inside one block, or overlaps two blocks. The first case corresponds to $0v_{a,i}$, and the second case of $0y_ax_b$ occurs if $ab$ is a factor of $w$, which corresponds to an edge $ab$ in the graph of $w$.
\end{proof}


We now proceed with the general case of morphisms which are not necessarily primitive. The following proposition basically states that there exists an algorithm for checking the condition on Parikh vectors of return words to $w_0$ from Theorem \ref{thdet1} (that they generate $\mathbb{Z}^\sigma$ as an additive group)  for words generated by morphisms:

\begin{proposition}\label{pr_alg}
The condition on Parikh vectors of return words to $w_0$ for words generated by morphisms from Theorem \ref{thdet1} is decidable.\end{proposition} 

\begin{proof} We provided an explicit algorithm which we divide into two parts. 

The first part of the algorithm decides, given an integer $m$, whether the Parikh vectors of the returns to zero modulo $m$  generate $(\mathbb{Z}/m\mathbb{Z})^\sigma$ as a group. It works as follows.
For each pair of letters $a,b \in \Sigma$ consider the set of vectors $S^i_{a, b}$, consisting of Parikh vectors modulo $m$ of all factors of $\phi^i(0)$ beginning with $a$ and ending with $b$. Suppose that some vector $x$ is new to $S^k_{a, b}$, i.e. $x \in S^k_{a, b}$ and $x \notin S^i_{a, b}$ for $i < k$. Since $x$ starts with $a$ and ends with $b$, we have $x=aub$ for some word $u$. Consider the first occurrence of $x$ in $\phi^k(0)$. Consider $\phi^k(0)$ as a concatenation of blocks $\phi(y)$ for $y\in \Sigma$, i.e. $\phi^k(0)=\phi(\phi^{k-1}(0))$. In such factorization, this occurrence of $x$ is a factor of a morphic image of some factor $v=cwd$ for some letters $c,d \in \Sigma$, i.e., $\phi(c)=rar'$, $\phi(d)=s'as$, $\phi(v) =raubs = \phi(c)\phi(w)\phi(d)$.  

Note that $V_w \in S^{k-1}_{c,d}$, since $\phi(w)$ is a factor of $\phi^k(0)$. 
It follows that $w$ is a new vector in $S^{k-1}_{c, d}$, since otherwise $aub$ would not be new for  $S^{k}_{a, b}$. Therefore, if for some $k$ we have $S^k_{a, b} = S^{k+1}_{a, b}$ for all pairs $a,b \in \Sigma$, then we also have  $S^{l}_{a, b} = S^{k}_{a, b}$ for each $ l > k$, for all pairs $a,b \in \Sigma$. Since each of these sets has a finite set of vectors, this will happen for some $k < \sigma^2 m^{\sigma}$. In particular, we will find a maximal $S^k_{0,0}$, which corresponds to the set of all returns to 0. 

 Now we describe the second part of the algorithm that checks, using the first part of the algorithm, whether the set of Parikh vectors of returns to $0$ generates $\mathbb{Z}^{\sigma}$ as a group. Consider an arbitrary prime $m$. If, using the algorithm, we determine that modulo this $m$ the set $S_{0,0}$ does not generate $(\mathbb{Z}/m\mathbb{Z})^\sigma$ as an additive group, then by \hyperlink{algn=allp}{Proposition \ref{algn=allp}} the entire set of returns also does not generate $\mathbb{Z}^\sigma$. Otherwise, $S_{0,0}$ has a basis
, and the matrix formed by these vectors has determinant non-zero in 
$\mathbb{Z}/m\mathbb{Z}$ by \hyperlink{algbas=det1}{Proposition \ref{algbas=det1}}. However, this implies that the determinant of the matrix composed of the Parikh vectors of these returns (now without reducing modulo 
$m$) is also non-zero. Let this determinant be equal to 
$k$. Then, for any prime $p$ coprime with $k$, this set of vectors generates the whole $(\mathbb{Z}/m\mathbb{Z})^\sigma$, whereas for the remaining primes dividing 
$k$ one can check this using the same algorithm. Again, by \hyperlink{algn=allp}{Proposition \ref{algn=allp}}, if for any of them the algorithm fails, then the set of returns does not generate $\mathbb{Z}^\sigma$ as an additive group.  Otherwise, it does.
\end{proof}

We will now show that in the binary case for recurrent words the condition of  Theorem \ref{thdet1} on returns always holds for morphisms with $\det A=\pm1$, 
which proves Theorem \ref{thdet1bin} in one direction. 

\hypertarget{thdet1>wbin}{}
\begin{thm} \label{thdet1>wbin}
    Let $w$ be an infinite recurrent binary word generated by a morphism $\phi$. If $\det A_\phi = \pm 1$, then $w$ satisfies the WELLDOC property.
\end{thm}

\begin{proof}
By Theorem \ref{thdet1}, it remains to prove that for each $p$ there are two distinct returns to $0$ such that their Parikh vectors generate  $(\mathbb{Z}/p\mathbb{Z})^\sigma$. 
    A return to $0$ is a binary word of the form $01\cdots 1$, hence we can take any two vectors distinct modulo $p$. We 
    claim that there are at least two distinct returns of this form.  Assume the converse, then such vector is unique modulo $p$; suppose it contains $x$ occurrences of $1$. If $00$ is a factor of $w$, then we have at least two returns to $0$ in $w$: $0$ and some other return beginning with $01$ (since by convention a binary word contains both $0$ and $1$). If $00$ is not a factor of $w$, then we claim that $11$ is a factor. 
    
    Suppose that $11$ is not a factor, then we have a periodic word $010101\cdots$.    
    We will now show that the morphism generating this word cannot have determinant $\pm 1$. The word $\phi(0)$ must be a prefix of $(01)^\infty$. If it is of the form $(01)^k$ for some $k \in \mathbb{N}$, then the determinant is divisible by $k$ and for $k \ne 1$ the determinant cannot be equal to $\pm 1$. If $\phi(0) = 01$, then 
    $\phi(1) = (01)^l$ for some $l \in \mathbb{N}$. In the case  $l \ne 1$, the determinant must be divisible by $l$, and in the case $l=1$, the determinant is equal to $0$. In the remaining case of $\phi(0) = (01)^k0$, we have $\phi(1) = 1(01)^l$ for some $k, l \in \mathbb{N}_0$. In this case $A_\phi = \left(\begin{smallmatrix} k+1 & l \\ k & l+1 \end{smallmatrix}\right)$, and $\det A_\phi=(k+1)(l+1) - kl = k+l+1 > 1$, since at least one of the numbers $k, l$ is not equal to $0$.
       
    So, we proved that $w$ contains the factor  $11$. We can write $\phi(0)$ in the following form: $$\phi(0)=0\underbrace{1\cdots 1}_{x+pk_1}0\underbrace{1\cdots 1}_{x+pk_2}\cdots0\underbrace{1\cdots 1}_{x+pk_{l-1}}0\underbrace{1\cdots 1}_{y+pk_l} \mbox{, for some }y, k_1, \ldots, k_l \in \mathbb{N}_0.$$ The remainders of the lengths of returns modulo $p$ are the same for all blocks of 1's, possibly except for the last one (since it does not have to be a return to 0 in $w$). The word $\phi(1)$ can also be rewritten in the form $$\phi(1)=\underbrace{1\cdots 1}_{z_1+pk'_1}0\underbrace{1\cdots 1}_{x+pk'_2}\cdots0\underbrace{1\cdots 1}_{x+pk'_{m-1}}0\underbrace{1\cdots 1}_{z_2+pk'_m}\mbox{, for some }z_1,z_2, k'_1, \ldots, k'_m \in \mathbb{N}_0.$$
    
    Since the word $w$ contains $01$ as a factor, it also contains $\phi(0)\phi(1)$ as a factor, so  $0\underbrace{1\cdots 1}_{y+pk_l}\underbrace{1\cdots 1}_{z_1+pk'_1}0$ is also a factor of $w$, hence $y+z_1 \equiv_p x$.
    
    Similarly, $10$ is a factor of $w$, so  $\phi(1)\phi(0)$ is a factor, and $11$ is a factor, so $\phi(1)\phi(1)$ is also a factor.  Thus, $z_2+0 \equiv_p x$ and $z_2+z_1 \equiv_p x$. 
    
    The last three equivalences modulo $p$ imply that $z_1 \equiv_p 0$, $z_2 \equiv_p x$ and $y \equiv_p x$. So Parikh vectors of all returns in the images of  $\phi(0)$ and $\phi(1)$ are the same modulo $p$. 

     We have that the word $\phi(0)$ is a concatenation of $l$ returns to 0, and  $\phi(1)$ is a concatenation of $m$ returns to 0, so $|\phi(0)|_0 = l$, $|\phi(0)|_1 = lx+ps$, $|\phi(1)|_0 = m$, $|\phi(1)|_1 = mx+pt$ for some $s, t \in \mathbb{N}_0$.  A contradiction with the condition on the determinant of the matrix: $\det A_\phi = l(mx+pt)-m(lx+ps) = p(lt-ms) \ne \pm 1$. 
\end{proof}

\hypertarget{setdet0}{}

\section{Necessity of the condition for the WELLDOC property} \label{secdet0}

In this section, we prove the necessity of the condition from Theorems \ref{thdet1bin} and \ref{thdet1all} for the WELLDOC property. For that, we will use the recognizability property of the morphisms from \cite{recogn}.

\subsection{Recognizable morphisms}

Let $w$ be an infinite word generated by a morphism $\phi$; we will consider its factorization into blocks $\phi(a)$ for letters $a \in \Sigma$. Let $f_w: \mathbb{N} \to \mathbb{N}$ be a function defined as follows:
\begin{equation*}
    i \mapsto f_w(i) = \begin{cases}
        |\phi(w[0, i))|, & \mbox{if } i>0,\\
        0, & \mbox{if }i=0.
    \end{cases}
\end{equation*}
Then $i$-th block $\phi(w_i)$ is $w[f_w(i), f_w(i+1))$. The indices $f_w(i)$ are called \textit{cutting points}.

We now recall some basic notation on symbolic dynamical systems arising from infinite words, which can be found, e.g., in \cite[\S 1.5]{W}. 
A \textit{biinfinite word} is a sequence of letters from $\Sigma$ indexed with $\mathbb{Z}$; the set of biinfinite words over $\Sigma$ is then denoted by  $\Sigma^\mathbb{Z}$. We equip the latter with the metrizable product topology, where $\Sigma$ is endowed with the discrete topology. Applying a morphism and cutting into blocks is defined similarly to one-way infinite words.

Let $T:\Sigma^\mathbb{Z} \to \Sigma^\mathbb{Z}$ be a left shift map: $(x_n)_{n \in \mathbb{Z}} \mapsto (x_{n+1})_{n \in \mathbb{Z}}$. A biinfinite word $x \in \Sigma^{\mathbb{Z}}$ is called periodic if $T^k(x) = x$ for some $k \ge 1$, and aperiodic otherwise. \textit{A subshift} $(X,T)$ is a set $X$ of infinite words which is shift-invariant (i.e., $T(X)=X$) and closed. We say that a subshift $(X,T)$ is generated by $x$ if $X$ is the set of all biinfinite words whose factors are also factors of the word $x$. 

A \textit{language} of a word $x$ is the set  $\mathcal{L}_x$ of all factors of $x$. For a language $\mathcal{L}$, we let   $(X_\mathcal{L},T)$ denote the subshift consisting of all biinfinite words whose factors are in $\mathcal{L}$. Finally, for a morphism  $\phi$ the language ${\mathcal{L}_\phi}$ is the language of words which are factors of words of the form $\phi^n(a)$ for some $a\in \Sigma$ and $n\in \mathbb{N}$. We will write $X_\phi$ instead of $X_{\mathcal{L}_\phi}$. 

\begin{defn}
   Let  $\phi$ be a morphism and $y \in \Sigma^\mathbb{Z}$ a biinfinite word. A pair $(k,x)$, where $k$ is an integer and $x$ is a biinfinite word, is said to be a {\it $\phi$-representation} of $y$  if $y = T^k(\phi(x))$. If $0 \le k < |\phi(x_0)|$, then $(k,x)$ is called a {\it central $\phi$-representation of $y$}. For a set $X \subset \Sigma^\mathbb{Z}$ of biinfinite words, we will say that {\it$\phi$-representation $(k, x)$ is from $X$} if $x \in X$. 
\end{defn}

\begin{defn}
    A morphism $\phi$ is called \textit{recognizable in  $X \subset \Sigma^\mathbb{Z}$} if each $y \in \Sigma^\mathbb{Z}$ has at most one central $\phi$-representation in $X$. A morphism $\phi$ is said to be \textit{recognizable in $X$ for aperiodic points} if each aperiodic word $y \in \Sigma^\mathbb{Z}$ has at most one central $\phi$-representation in $X$. 
\end{defn}
    \begin{defn}
    Let $\phi$ be a morphism, $w$ a biinfinite word, and $v=\phi(w)$. A morphism $\phi$ is called \textit{recognizable in the sense of Moss\'e for $w$} if there exists $l$ such that for each $m=f_w(i)$ and $m' \in \mathbb{Z}$ the equality $v[m-l, m+l]=v[m'-l, m'+l]$ implies $m'=f_w(j)$ for some $j$.   
\end{defn}

We will use Theorems 5.3 and 2.5 from \cite{recogn}.

\begin{thm}[Theorem 5.3 in \cite{recogn}] \label{thallrec}
    Let $\phi$ be a morphism. Then it is recognizable in $X_\phi$ for aperiodic points. 
\end{thm}

\begin{thm}[Theorem 2.5 in \cite{recogn}] \label{threcMos}
    Let $\phi$ be a morphism, $x$ a biinfinite word, $(X,T)$ a subshift generated by $x$. If $\phi$ is recognizable in $X$, then $\phi$ is recognizable in the sense of Moss\'e for $x$.
\end{thm}

\subsection{Proof of necessity of the condition for the WELLDOC property}
In this subsection we give a proof of the necessity of the condition from our WELLDOC criterion. 

\begin{defn}
    Let $w$ be an infinite word generated by a morphism. A factor $u$ of the word $w$ is called {\it cutting} if there exists an index $k$ such that for each occurrence of $u$ in $w$ at position $m$ the index $m+k$ is a cutting point.  
\end{defn}

Our goal will be to find a cutting factor due to the following theorem:

\hypertarget{thinWD}{}

\begin{thm} \label{thinWD}
    Let $w$ be an infinite word generated by a morphism $\phi$ such that $\det A_\phi \ne \pm 1$. If $w$ contains a cutting factor $u$, then $w$ does not satisfy the WELLDOC property. 
\end{thm}

\begin{proof} Consider the set of prefixes preceding $u$; let $p$ be one of them. Then $pu$ is also a prefix of $w$. By the definition of a cutting factor, the index $|p|+k$ is a cutting point, where $k$ is the constant associated with this cutting factor. So, $pu[0,k) = \phi(s)$ for some word $s \in \Sigma^*$. This means that $V_p$ can be expressed as $A_\phi V_{s} - V_{u[0,k)}$.
Since  $\det A_\phi \ne \pm 1$, we cannot express all vectors in ($\mathbb{Z}/m\mathbb{Z})^\sigma$ this way for some $m$ given by Proposition \ref{algdet0}.
This contradicts the WELLDOC property. 
\end{proof}

To apply Theorems \ref{thallrec} and \ref{threcMos}, we use the fact that each recurrent right-infinite word can be continued to biinfinite word with the same set of factors. Although this fact is likely to be folklore, we provide its proof for the sake of completeness.

\begin{lm} \label{lmcontoleft}
    Let $w$ be a right-infinite recurrent word. Then there exists a biinfinite word $w'$ such that $w'_i = w_i$ for each $i \ge 0$ and the set of factors of $w'$ coincides with the set of factors of $w$.
\end{lm}

\begin{proof} 
    For the proof, we build a tree of prolongations to the left of the word $w$ as follows. Starting with an empty tree, at step $k$ we build a tree of height $k$: its root is the prefix of $w$ of length $k$, and the vertices of its branches are marked with letters, so that each branch corresponds to some factor of $w$ of length $2k$ with suffix $w[0,k)$. The tree on step $k$ is obtained by extending the tree on step 
     $k-1$ as follows. For each branch of the tree (corresponding to a word $u$ of length $k-1$ such that $uw[0,k-1)$ is a factor of $w$), we add to its leaf children marked with each letter $a\in\Sigma$ such that $auw[0, k)$ is a factor of $w$. 
     In fact, for  each $k$, the tree contains all factors of $w$ of length $2k$ with suffix $w[0,k)$. 

     Let us now prove that the tree is infinite. 
    Since $w$ is recurrent, each its prefix occurs in it infinitely often. On step $k$ we build all possible extensions to the left of length $k$ of the prefix of length $k$ of $w$. Since this prefix of length $k$ occurs infinitely often in $w$, it has infinitely long extensions to the left in $w$. Therefore, at each step at least one branch is continued.

    By K\"{o}nig's lemma, each infinite tree with bounded degrees of vertices contains an infinite path. We proved that our tree is infinite, and the number of children of each vertex is bounded by the size of the alphabet. So, our tree contains an infinite path, corresponding to an infinite prolongation to the left. The word $w'$ can be chosen as the word $w$ with this infinite prolongation to the left;  the set of its factors coincides with the set of factors of $w$ by construction. 
\end{proof}



\begin{thm} \label{thcutexists}
    Let $w$ be a right-infinite recurrent aperiodic word generated by a morphism $\phi$. Then there exists $k$ such that each factor $u$ of length greater than $k$ is a cutting factor.
\end{thm}

\begin{proof}
    Let $w'$ be a biinfinite word with the same set of factors as $w$ and having $w$ as a suffix (see Lemma \ref{lmcontoleft} for an explicit construction). By Theorem \ref{thallrec}, the morphism $\phi$ is recognizable in $X_\phi$ for aperiodic points. By the definition of $X_{\phi}$, the subshift $(X,T)$ generated by $w'$ is contained in $(X_\phi,T)$. Therefore, $\phi$ is recognizable in the sense of Moss\'e for $w'$ by Theorem \ref{threcMos}.  Let $l$ be as in the definition of recognizability in the sense of Moss\'e for $w'$.
    
    Let $z$ be any factor of length greater than $2l+||\phi||$, where $||\phi||$ is the maximum length of the images of letters in $\phi$, i.e., $||\phi|| = \max \limits_{a \in \Sigma} (|\phi(a)|)$. Consider an occurrence of this factor in the word $w$, denote the position of this occurrence by $i$, so that $z = w[i, i+|z|)$. It has a cutting point in the interval of indices $[i+l, i+l+||\phi||)$, say at some index $i+l+n$, where $n < ||\phi||$. Now take any other occurrence of $z$ in $w$; denote its position by $j$, so that $z = w[j, j+|z|)$. Since $\phi(w')$ has the same letters at positive positions, we have
    $$\phi(w')[(i+l+n)-l, (i+l+n)+l] = z[k,2l+n] = \phi(w')[(j+l+n)-l, (j+l+n)+l]ю$$ 
    Since $\phi$ is recognizable in the sense of Moss\'e, the point $m+l+n$ is also a cutting point. 
    This is exactly the definition of a cutting factor. Thus, we have proved the theorem for $k=2l+||\phi||$.
\end{proof}

Now we prove the necessity of the condition for the WELLDOC property for words generated by morphisms on an alphabet  $\Sigma$ with $|\Sigma| \ge 2$.

\begin{thm}\label{th:necessary}
     Let $w$ be an infinite word over an alphabet $\Sigma$ with $|\Sigma| \ge 2$ generated by a morphism $\phi$ such that $\det A_\phi \ne \pm 1$. Then the word does not satisfy the WELLDOC property. 
\end{thm}

\begin{proof}
    If $w$ is purely periodic, then the WELLDOC property does not hold for its minimal period when the period length exceeds 1.  
    If its length is equal to $1$, then $w=0^\infty$, which is a one-letter word, and we assumed that the alphabet size is at least 2.

    If $w$ is non-recurrent, i.e., there exists a factor $u$ that occurs only finitely many times, then the set  $X_u$ is finite, so the WELLDOC property does not hold, for example, for any $m > |X_u|$. 
    
    Otherwise, if $w$ is aperiodic and recurrent, then $w$ has a cutting factor by 
    Theorem~\ref{thcutexists} 
    and hence, by Theorem~\ref{thinWD}, the WELLDOC  property does not hold.
\end{proof}

\section*{Acknowledgements}

This work was supported by the Russian Science Foundation, project 25-21-00535. The authors are grateful to Wolfgang Steiner for recommending to consult the paper~\cite{recogn}.

\end{document}